\documentclass[12pt]{amsart}

\usepackage{amsmath,amsfonts,amssymb,cite,hyperref}
\usepackage{graphicx, epstopdf}
\usepackage{color,verbatim}

\newcommand{\RR}{\mathbb R}

\newcommand{\ZZ}{\mathbb Z}

\DeclareMathOperator{\conv}{conv}
\DeclareMathOperator{\Cay}{Cay}
\DeclareMathOperator{\Span}{span}

\DeclareMathOperator{\lft}{lift}
\newcommand{\lift}[1]{\lft(#1)}

\def\D{{\mathcal D}}

\def\Z{{\mathbb Z}}
\def\R{{\mathbb R}}

\newtheorem{theorem}{Theorem}[section]
\newtheorem{lemma}[theorem]{Lemma}
\newtheorem{corollary}[theorem]{Corollary}
\newtheorem{proposition}[theorem]{Proposition}
\newtheorem{conjecture}[theorem]{Conjecture}
  \theoremstyle{definition}
\newtheorem{definition}[theorem]{Definition}
\newtheorem{example}{Example}

\begin{document}

\author{Ngoc Mai Tran}
\address{Department of Mathematics, University of Texas at Austin, Texas TX 78712, USA and 
Hausdorff Center for Mathematics, Bonn 53115, Germany}
\email{ntran@math.utexas.edu}

\author{Josephine Yu}
\address{School of Mathematics, Georgia Institute of Technology, Atlanta GA 30332, USA}
 \email{jyu@math.gatech.edu}

\begin{abstract}
In a recent and ongoing work, Baldwin and Klemperer explored a connection between tropical geometry and economics. They gave a sufficient condition for the existence of competitive equilibrium in product-mix auctions of indivisible goods.  This result, which we call the Unimodularity Theorem, can also be traced back to the work of Danilov, Koshevoy, and Murota in discrete convex analysis. We give a new proof of the Unimodularity Theorem via the classical unimodularity theorem in integer programming. 
We give a unified treatment of these results via tropical geometry and formulate a new sufficient condition for competitive equilibrium when there are only two types of product. Generalizations of our theorem in higher dimensions are equivalent to various forms of the Oda conjecture in algebraic geometry.
\end{abstract}
\title{Product-Mix Auctions and Tropical Geometry}
\maketitle


\section{Introduction.}
An auction is a mechanism to sell goods.  It is a process that takes in a collection of bids and produces a set of winners, the amount of goods each winner gets, and the prices the winners have to pay. A product-mix auction is a sealed bid, static auction for differentiated  goods. That is, no one knows others' bids, each agent submits the bids only once, and there are different types of goods.
 The goods are said to be {\em indivisible} if they can only be sold in integer quantities, and this is the case that interests us. Each agent (bidder) indicates her valuation, or how much each {\em bundle} of goods is worth to her. After seeing the bids, the economist sets the per-unit price for each type of good, splits up the supply into possibly empty bundles, one for each agent. Each agent pays for her assigned bundle according to the set prices. A {\em competitive equilibrium} exists if there is a choice of prices and a way to exactly divide up the supply so that each agent is assigned a bundle that maximizes her utility. Competitive equilibrium is seen as a desirable, satisfaction-for-all outcome. Thus, the central economic question is to find conditions under which competitive equilibrium always exists, and if possible, is easy to compute. 

The product-mix auction was first proposed by Klemperer \cite{klemperer08, klemperer2010product} at the heels of the 2007 Bank of England crisis. As a game, its strength is in its simplicity. In particular, if competitive equilibrium exists, then agents who want to maximize their utility should bid their true valuation of the bundles, and this is a desirable economic outcome. The product-mix auction subsumes many models as special cases \cite{BaldwinKlemperer}, while itself being a special case of multi-unit combinatorial auction \cite{de2003combinatorial}.
Thus, the literature related to product-mix auction spreads across several fields, from economics to optimization to discrete convex analysis. Each field has its own language, central questions and notations, and consequently, discoveries are often repeated. Competitive equilibrium in product-mix auction is an instance of general equilibrium, more specifically, Walrasian equilibrium. This classical economic model dates back to Walras, a French mathematician in the late 19th century. However, in a product-mix auction with indivisible goods, convex analysis and fixed point techniques often used in general equilibrium theory do not immediately apply. 

Danilov, Koshevoy and Murota took this view in considering Walrasian economies with indivisible goods \cite{DKM} and searched for conditions where the discrete analogue of classical fixed point arguments still work. They obtained the Unimodularity Theorem \cite[Theorem 4]{DKM} in 2001, although its proof appears three years later as a special case of \cite[Theorem 3]{danilov2004discrete}. A version of this theorem surfaces independently in the algebraic geometry community, with Howard \cite{howard} rediscovering what is essentially \cite[Theorem 3]{danilov2004discrete} in 2007. Then in 2012, the Unimodularity Theorem was once again independently discovered by Baldwin and Klemperer in~\cite{BaldwinKlemperer12} in the form and context stated here, whose tropical geometric proof we give in Section~\ref{sec:unimodularity} below employs the same key argument as \cite{howard, danilov2004discrete}. Since then, Baldwin and Klemperer have used tropical geometry to discover other sufficient conditions for competitive equilibrium in product-mix auctions, most notably the Intersection Count Theorem \cite[Theorem 5.12]{BaldwinKlemperer}. Their novel lattice-counting approach is complementary to the results discussed here, see \cite{BaldwinKlemperer} for further developments.

Our paper has two contributions. First, we rewrite product-mix auction as a particular integer program, cast competitive equilibrium as a linear versus integer programming question, and show how the Unimodularity Theorem follows from the usual unimodularity theorem in integer programming. It differs from the proof of \cite{DKM,danilov2004discrete,BaldwinKlemperer,howard}, where one needs to find classes of polytopes in which the operations of taking Minkowski sums and taking integer points commute. Second, we give a unified treatment of competitive equilibrium in product-mix auctions, with self-contained proofs from all the perspectives considered in the literature: discrete convex analysis, tropical geometry, and integer programming. Utilizing connections to lattice polytopes, we give a new family of product-mix auctions where competitive equilibrium is guaranteed to hold (cf. Theorem~\ref{thm:n2}). Generalizations of this result are  equivalent to various forms of the Oda Conjecture in algebraic geometry. Our setup of Theorem \ref{thm:n2} is as far as possible from transverse intersections, which are the focus of the developments in \cite{BaldwinKlemperer}, and thus appears to be new.  

Unlike \cite{BaldwinKlemperer}, whose audience mainly consists of economists, this paper is an exposition on product-mix auctions for a mathematical audience with minimal prior knowledge on auction theory. Following \cite{BaldwinKlemperer}, we choose to state the product-mix auction in terms of tropical geometry, as it introduces an elegant language and key geometric insights. The paper is organized as follows. After defining the product-mix auctions in~\S\ref{sec:auctions}, we make connections to tropical geometry in~\S\ref{sec:tropical} and give a simple proof to the Unimodularity Theorem in~\S\ref{sec:unimodularity}. Our integer programming approach to competitive equilibrium appears in \S\ref{sec:LPIP} where we give a new proof of the Unimodularity Theorem.  We discuss other integer programming formulations in \S\ref{sec:lpip.cayley}, and connections to subset sum  in \S\ref{sec:polytime}. Finally in \S\ref{sec:oda} we make connections to the study lattice polytopes arising from toric geometry, such as the Oda conjecture.

\vskip12pt
\noindent \textbf{Notations.} For a set $A \subset \R^n$, let $\conv(A)\subset \R^n $ denote its convex hull, $\conv_\Z(A) = conv(A) \cap \Z^n$ denote its integer convex hull, $|A|$ denote its cardinality. 
\vskip12pt \noindent
\textbf{Acknowledgments.} 
This project came together thanks to a workshop organized by Simona Settepanella in 2014, and a talk request by Matt Baker. We thank Elizabeth Baldwin, Anton Leykin, Bernd Sturmfels, Gleb Koshevoy and Rakesh Vohra for interesting discussions, Luis Bosc\'an, Paul Klemperer, Charles Wang, and especially Elizabeth Baldwin for comments on earlier versions, and the anonymous referees for their careful readings and helpful suggestions. Ngoc Tran was supported by an award from the Simons Foundation (\#197982 to the University of Texas at Austin) and the Bonn Junior Fellowship at the Hausdorff Center for Mathematics. Josephine Yu was supported by NSF DMS grants \#1101289 and \#1600569.

\section{Product-Mix Auctions.}
\label{sec:auctions}

Suppose there are $n$ types of indivisible goods. A \emph{good bundle} is a point in $\Z^n$. For $j = 1, \ldots, J$, the $j$-th agent gives the economist her valuation function (bids)
$$
u^j: A^j \to \R, 
$$
where $A^j \subset \ZZ^n$ and $u^j(a)$ is her bid for bundle $a \in A^j$. Here negative coordinates (with possibly negative valuations) mean that the agent wants to sell the products.  In this setup the sellers and buyers play the same role. 

The Minkowski sum 
$$A = \sum_{j=1}^JA^j := \left\{\sum_{j=1}^J a^j \mid a^j \in A^j \text{ for each } j=1,\dots,J \right\}$$ 
is the set of all possible good bundles in the economy that can potentially be matched to this set of agents. For a fixed price vector $p \in \R^n$, where $p_i$ is the price for one unit of the $i$-th good, the {\em utility} of the $j$-th agent buying bundle $a$ at price $p$ is $u^j(a)-p \cdot a$.
The {\em demand set} of the $j$-th agent at price $p$ is the set of bundles that maximizes her utility
\begin{equation}
\label{eq:profit}
D_{u^j}(p) := \arg\max_{a \in A^j}\{u^j(a) - p\cdot a\}.
\end{equation}
The aggregate valuation function $U: A \to \R^n$ is the maximum total valuation taken over all ways to partition each bundle $a \in A$:
\begin{equation}\label{eqn:U}
U(a) := \max\left\{\sum_{j=1}^J u^j(a^j) \mid a^j \in A^j \text{ and } \sum_{j = 1}^J a^j = a \right\}. 
\end{equation}
The aggregate demand at $p$ is the demand set of the aggregate valuation $U$
$$ D_U(p) :=  \arg\max_{a \in A} \{ U(a) - p \cdot a\}.$$  
One can check that it is the Minkowski sum of the individual demands
$$ D_U(p) = \sum_{j=1}^J D_{u^j}(p) \subseteq A. $$
The input of a product-mix auction of indivisible goods consists of the valuation functions $\{u^j \mid j = 1, \ldots, J\}$ and a fixed supply bundle $a$. 
The economist needs to find a price vector $p \in \R^n$ so that $a \in D_U(p)$. If such a price exists, we say that a {\em competitive equilibrium exists at~$a$}.

\begin{definition}
The set $\{u^j\}$ of valuation functions has \emph{competitive equilibrium} at $a\in\conv_\Z(A)$ if there exists a price $p \in \R^n$ such that $a \in D_U(p)$.  We say that \emph{competitive equilibrium exists} if it exists for every $a \in \conv_\Z(A)$.
\end{definition}

When the valuations are all constant, competitive equilibrium exists at $a \in \conv_\Z(A)$ if and only if $a \in A_1 + \cdots + A_J$.  When each $A_i \subset \ZZ$, checking the existence of competitive equilibrium is exactly the SUBSET-SUM problem, which is NP-complete.  However, in this case of $n=1$, competitive equilibrium exists for constant valuation if each $A_i$ contains all lattice points in its convex hull.  This is not true for $n \geq 2$.  For example, $A_1 = \{(0,0), (1,1)\}$ and $A_2 = \{(1,0), (0,1)\}$ contain all the lattice points in their respective convex hulls, but competitive equilibrium does not exist at $(1,1)$ under any valuation, including constant valuation.

Economists are greatly interested
in general conditions on $\{u^j\}$ (or $\{u^j, a\}$) for the existence of competitive equilibrium 
and algorithms for finding prices $p$ and the winner assignment at those prices. If there is one seller and $J-1$ buyers, in general, competitive equilibrium does not guarantee maximal profit for the seller. The seller may make a bigger revenue by selling fewer products, see Example~\ref{ex:max.profit}. Maximizing profit for the seller is a different problem from finding competitive equilibrium. Often profit maximization is the objective of combinatorial auctions \cite{de2003combinatorial}, and thus results in that literature do not transfer immediately.

\medskip
\begin{example}\label{ex:max.profit}
Suppose $n=1$ and  the seller has two objects of the same type to sell.  There are two agents, $A^1 = A^2 = \{0,1,2\}$ with valuations 
\begin{align*}
u^1(0) & =0, ~~u^1(1)=10,~~ u^1(2)=11,\\
u^2(0) & =0, ~~u^2(1)=2,~~u^2(2)=3.
\end{align*}
The competitive equilibrium prices are $\{p \mid 1 \leq p \leq 2\}$ where each agent buys one item,  so the maximal revenue under competitive equilibrium for the seller receives is $4$.  However, if the agent sets price $p = 10$, then agent one buys one object, agents two buys nothing. Not all items are sold, so competitive equilibrium does not hold at $(1,1)$, but the agent achieves a larger revenue of $10$.
\qed
\end{example}

\smallskip

We will demonstrate the failure of competitive equilibrium in two examples. Example \ref{ex:first} shows that if $A \subsetneq \conv_\Z(A)$, competitive equilibrium can be sure to fail regardless of the agents' valuations $u^j$'s. This is because the aggregated valuation $U$ is only finite on $A$ (and is negative infinity elsewhere), so competitive equilibrium does not hold for points outside of $A$. On the other hand, Example \ref{ex:1} shows that $A = \conv_\Z(A)$ is only necessary, but not sufficient, for competitive equilibrium to hold.

\medskip
\begin{example}\label{ex:first}
Let $A^1 = \{(0,0), (1,0)\}$, $A^2 = \{(0,0), (1,2)\}$. For any valuation functions $u^1, u^2$ with domains $A^1, A^2$ respectively, the domain of the aggregate valuation $U$ is the Minkowski sum $A = A^1 + A^2 = \{(0,0), (1,0), (1,2), (2,2)\}$. 
But $\conv_\Z(A)$ contains the point $(1,1)$ while $A$ does not. Since $(1,1)$ is not in the domain of~$U$, $U((1,1)) = -\infty$. Thus, competitive equilibrium on $\conv_\Z(A)$ would fail at the point $(1,1)$ for any valuations $u^1$ and $u^2$. 
\end{example}

\begin{figure}
\begin{center}
\includegraphics[scale=1]{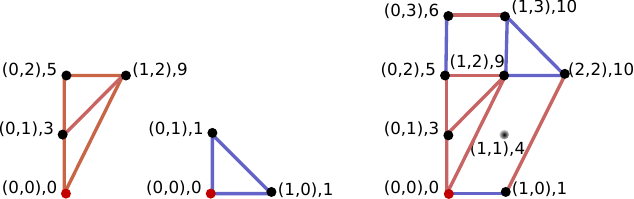}
\end{center}
\caption{A product-mix auction with two agents and two product types. From left to right: sets $A^1,A^2,A$, shown with valuations $u^1, u^2, U$ next to the corresponding bundles. The polyhedral complex is the regular subdivision $\Delta_{u^1}, \Delta_{u^2}, \Delta_U$, as explained in Example \ref{ex:subdiv}. Although the point $(1,1)$ is in the set $A_1+A_2$, there is no price at which the bundle $(1,1)$ is demanded, so competitive equilibrium fails at that point. See Example~\ref{ex:1}.}
\label{fig:subdivision}
\end{figure}

\begin{figure}
\begin{center}
\includegraphics[scale=1.2]{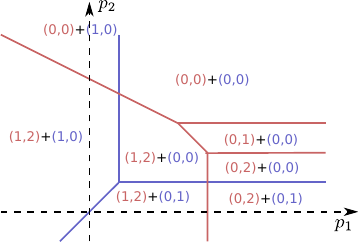}
\end{center}
\caption{The price space $\RR^2$ is partitioned according to the demand sets of the agents from Figure \ref{fig:subdivision}, see also Example~\ref{ex:1}. The regions are labeled with the demanded bundles from each of the two agents.  Note that none of the aggregate demands is equal to $(1,1)$.  If a price falls on the boundary of multiple regions, then there are more than one bundles in the demand set.  The boundary between regions is the negative of the union of tropical hypersurfaces; note the negative sign in~\eqref{eq:profit} compared to~\eqref{eq:tropPoly}.} 
\label{fig:prices}
\end{figure}

\medskip
\begin{example}
\label{ex:1}
Even when $A = \conv_\Z(A)$, competitive equilibrium may fail to exist. Suppose we have two types of goods and two agents, with valuation functions shown in Figure~\ref{fig:subdivision}.  Although the bundle $(1,1)$ is in $A$, the aggregate valuation at $(1,1)$ is so small that there is no price at which $(1,1)$ is in the aggregate demand set.  
\end{example} 

\section{Tropical mathematics.}
\label{sec:tropical}

We now restate the problem of determining competitive equilibrium in product-mix auctions in terms of tropical geometry. This formulation was first noted in the inspirational paper of Baldwin and Klemperer \cite{BaldwinKlemperer12}. For an extensive introduction to tropical geometry, see the monograph of Maclagan and Sturmfels \cite{maclagan2015introduction}. As a language, tropical geometry gives a clean statement and proof of the Unimodularity Theorem (cf. Theorem \ref{thm:bk}). Furthermore, it allows easy analysis of known families product-mix auctions for which competitive equilibrium is guaranteed to exist. We demonstrate such a family in Example \ref{ex:murota}.

Here we focus on the max-plus semiring $(\R\cup\{-\infty\}, \oplus,\odot)$, a set with tropical addition~$\oplus$ and tropical multiplication~$\odot$ such that
$$
r \oplus s := \max(r,s), \text{ and } r \odot s := r + s \mbox{ for all } r,s \in \RR\cup\{-\infty\}.
$$
For our purpose, any function $u : A \subset \ZZ^n \rightarrow \RR$ defines three important objects: the tropical Laurent polynomial $f_u$, the tropical hypersurface $T(f_u)$, and the regular subdivision $\Delta_u$ of $A$. If the function $u$ is the valuation function, then these three objects are also called the utility, the locus of indifference prices (LIP), and the demand complex  respectively in the economics literature \cite{BaldwinKlemperer}.

Let us elaborate. The \emph{tropical Laurent polynomial} $f_u$ with powers in $A \subset \ZZ^n$ and coefficients $u(a) \in \RR$ for $a \in A$ is
\begin{equation}
\label{eq:tropPoly}
f_u(x) = \bigoplus_{a \in A} u(a) \odot x^{\odot a} 
\end{equation}
which is equal to the function $x \mapsto \max \{u(a) + a \cdot x \mid a \in A\}$. 

The {\em tropical hypersurface} $T(f_u)$ of a tropical polynomial $f_u$ in $n$ variables is defined as the set of $x \in \R^n$ where the set $\arg\max \{u(a) + a \cdot x \mid a \in A\}$ has cardinality at least two. That is, the maximum the expression of $f_u(x)$ is achieved in at least two terms. In the usual algebra, $f_u: \RR^n \to \R$ is a piecewise-linear convex function, and $T(f_u)$ is the locus where the function $f(x)$ is not smooth. The regular subdivision $\Delta_u$ of $A$ induced by $u$ is defined as
$$\Delta_u := \{ \arg\max_{a \in A} [u(a) \odot x^{\odot a}]   \mid x \in \RR^n \},$$
which is a collection of subsets of $A$.
The elements of $\Delta_u$ are called {\em cells} of the subdivision.  The convex hull of a cell is called a {\em face}.  Note that a cell is a finite set of integer points while a face is a polytope. A point $a \in A$ is called a {\em marked point} or a {\em lifted point}  of the subdivision if it is contained in one of the cells. Every point in $A$ is contained in a face, but not necessarily in a cell, of the subdivision.

The regular subdivision $\Delta_u$ is {\em dual} to the tropical hypersurface $f_u$ in the sense that there is a natural bijection between positive dimensional faces of $\Delta_u$ and the faces of the tropical hypersurface, as follows.
Given the tropical polynomial $f_u$ as above, consider the following equivalence relation on $\RR^n$: two points $x, x' \in \RR^n$ are equivalent if \[\arg\max\{u(a) + a \cdot x \mid a \in A\} = \arg\max\{u(a) + a \cdot x' \mid a \in A\}.\]  This gives a partition of $\RR^n$ where each equivalence class is a relatively open convex polyhedron, defined by linear equations and strict linear inequalities.  The tropical hypersurface is the union of the equivalence classes of dimension $< n$ in this partition.  The correspondence \[x \leftrightarrow \arg\max\{u(a) + a \cdot x \mid a \in A\}\] is a bijection between the equivalence classes in $\RR^n$ and faces of $\Delta_u$.

When all the coefficients $u(a)$ of the tropical polynomial $f_u = \oplus_{a \in A} u(a) \odot x^{\odot a}$ are equal to $0$ (or any other fixed constant), then the tropical hypersurface depends only on the Newton polytope $P = \conv A$.  On the other hand, given an integral polytope $P$, let $f$ be a tropical polynomial with constant coefficients whose Newton polytope is $P$.  We define $T(P)$ to be the tropical hypersurface $T(f)$, which is a subfan of the normal fan of $P$ consisting of the normal cones to the faces of positive dimension in $P$.

The connection to product-mix auctions is clear upon comparison with the equations in the previous section. Plugging in $x = -p$, $f_u(-p)$ is the maximal utility of an agent with valuation $u$ at price~$p$, $-T(f_u)$ is the set of prices where the agent is indifferent between two bundles, and the cells in $\Delta_u$ are precisely all the possible demand sets $D_u(p)$ as we vary the price $p$, that is, $\Delta_u = \{D_u(p) \mid p \in \R^n\}$. More importantly, both aggregation over multiple agents and competitive equilibrium are simple to describe in tropical terms, as follows.

\begin{lemma}\label{lem:simple.operations}
Let $u^1, \dots, u^J$ be valuation functions of $J$ agents on supports $A^1,\dots,A^J \subset \ZZ^n$ respectively. Let $U$ be their aggregate valuation function on $A = \sum_{j=1}^J A^j$. The following hold
\begin{itemize}
  \item $f_U = f_{u^1} \odot \ldots \odot f_{u^J}$,
  \item $T(f_U) = \bigcup_{j=1}^J T(f_{u^j})$, and
  \item $D_U(p) = D_{u^1}(p) + \ldots + D_{u^J}(p)$ for any price vector $p \in \RR^n$
  \item Competitive equilibrium exists at a point $a \in \Z^n$ if and only if it is a marked point in $\Delta_U$.
\end{itemize}
\end{lemma}

\medskip
\begin{example}\label{ex:subdiv}
Revisit Example \ref{ex:1}. Figures~\ref{fig:subdivision} shows the regular subdivisions $\Delta_{u^1}$ on the far left, $\Delta_{u^2}$ in the middle, and the regular mixed subdivision $\Delta_U$ in the far right. Figure \ref{fig:prices} shows the negatives of the tropical hypersurfaces $T(f_{u^1})$ in blue and $T(f_{u^2})$ in red. Their set union is negative of the tropical hypersurface $T(f_U)$. The point $(1,1)$ is not marked, and thus competitive equilibrium fails. 
\end{example}

We give a geometric interpretation of when a given point in a regular subdivision $\Delta_u$ is marked (or lifted). Visualize a function $u: A \subset \Z^n \to \R$ as lifting each point $a \in A$ to a height $u(a)$ in a new dimension, producing the graph of the function $u$
\begin{equation}\label{eqn:lift.A}
\lift{A} := \{(a, u(a)) \mid a \in A \} \subset A \times \RR.
\end{equation}
Let $\widetilde{u}$ be the concave majorant of $u$, that is, the smallest concave function on $\conv(A)$ such that $\widetilde{u}(a) \geq u(a)$ for all $a \in A$.  Think of $\widetilde{u}$ as a piecewise-linear surface formed by extending a cling wrap over $\lift{A}$. The projection linear pieces of $\widetilde{u}$ onto $\conv(A)$ are the faces of $\Delta_u$, the regular subdivision of $A$ induced by $u$. A point $a \in \Z^n$ is {\em lifted} (or marked) if and only if $\widetilde{u}(a) = u(a)$, that is, it is literally lifted high enough by $u$ so that it touches the cling wrap $\widetilde{u}$. With this view, the following lemma is straightforward from the definitions. 
\begin{lemma}
\label{lem:CEconditions}
Let $u^1, \dots, u^J$ be valuation functions of $J$ agents on supports $A^1,\dots,A^J \subset \ZZ^n$ respectively.  Let $U$ be their aggregate valuation function on $A = \sum_{j=1}^J A^j$.  The following are equivalent
\begin{enumerate}
  \item Competitive equilibrium exists.
  \item $A = \conv_\Z(A)$ and $U = \widetilde{U}$ on $A$.
  \item For every $p \in -T(f_U)$, $D_U(p) = \conv_\Z(D_U(p))$.
  \item For every vertex $p \in -T(f_U)$, $D_U(p) = \conv_\Z(D_U(p))$.
\end{enumerate}
\end{lemma}
Computationally, the second condition says that checking competitive equilibrium involves evaluating $U$ and $\widetilde{U}$ at all points in $A$. The last condition says that this can also be done by checking a finite set of prices, namely, those which appear as the negative of the vertices of the tropical hypersurface $T(f_U)$. This is because such prices define support vectors for the inclusion-maximal faces of $\Delta_U$. If all lattice points in all maximal faces of $\Delta_U$ are marked, then all points of $A$ are marked in $\Delta_U$.  

We say that $u: A \subset \Z^n \to \R$ is {\em concave} if $A = \conv_\Z A$ and $u = \tilde{u}$ on $A$. We conclude with following Lemma \ref{lem:poly1}, a characterization for competitive equilibrium on a given face. This Lemma shows that competitive equilibrium is equivalent to checking whether Minkowski sum commutes with taking lattice points.  Danilov and Koshevoy used this in their proof of the Unimodularity Theorem \cite{danilov2004discrete, DKM}. Such questions also appear in connections to toric varieties~\cite{Toric}. We build on this link in Section \ref{sec:oda}, employing theorems from algebraic geometry to obtain new theorems and conjectures on competitive equilibrium. See Baldwin and Klemperer \cite[Theorem 5.16, Theorem 5.21]{BaldwinKlemperer} for other criterion for competitive equilibrium on a given face via lattice counting approaches.

\begin{lemma}\label{lem:poly1}
Let $u^1, \dots, u^J$ be concave valuation functions of $J$ agents, let $U$ be their aggregate valuation function. For $p \in \R^n$, competitive equilibrium holds at all $a \in D_U(p)$ if and only if
\begin{equation}\label{eqn:normal}
\conv_\Z(D_U(p)) = \conv_\Z(D_{u^1}(p)) + \ldots + \conv_\Z(D_{u^J}(p)).
\end{equation}
 \end{lemma}
\proof{Proof.}
By Lemma \ref{lem:simple.operations}, $D_U(p) = D_{u^1}(p) + \ldots + D_{u^J}(p)$. Since the valuations $u^j$ are concave, $\conv_\Z(D_{u^j}(p)) = D_{u^j}(p)$ for $j = 1, \ldots, J$. 
Thus, 
$$ D_U(p) = \conv_\Z(D_{u^j}(p)) + \ldots + \conv_\Z(D_{u^j}(p)). $$
By Lemma \ref{lem:CEconditions}, competitive equilibrium holds at all $a \in D_U(p)$ if and only if $\conv_\Z(D_U(p)) = D_U(p)$. Thus, such competitive equilibrium holds if and only if (\ref{eqn:normal}).
\hfill\qed
\endproof

\section{The Unimodularity Theorem.}
\label{sec:unimodularity}
We are now ready to state and prove the Unimodularity Theorem as formulated by Baldwin and Klemperer \cite{BaldwinKlemperer12} in 2012. Another version was first discovered Danilov, Koshevoy and Murota \cite[Theorem 4]{DKM} in 2001, motivated by Walrasian economies. Both proofs, presented in this section, rely on the same key result, Lemma \ref{lem:unimodular}, which itself was independently discovered by Howard \cite{howard} in 2007. We will present a different proof of the Unimodularity Theorem in \S\ref{sec:LPIP}.

A non-zero integer vector is called {\em primitive} if the greatest common divisor of its coordinates is~$1$. A set of vectors $\D \subset \ZZ^n$ is called {\em unimodular} if every linearly independent subset of $n$ vectors in $\D$ spans $\ZZ^n$ over $\Z$.

\smallskip
\begin{definition}
Let $\D$ be a set of primitive integer vectors in $\Z^n$ and $A \subset \Z^n$. We say that a valuation  $u: A \to \R$ is of {\em demand type} $\D$ if every edge of the subdivision $\Delta_u$ is parallel to a vector in $\D$. 
\end{definition}
\smallskip
 
In other words, $u$ is of demand type $\D$ if all the integer facet normals of $T(f_u)$ lie in $\D$. Since the tropical hypersurface $T(f_U)$ is the union of $T(f_{u^j})$ for all $j = 1,\dots,J$, it follows that the aggregate valuation $U$ is of demand type $\D$ if and only if $u^j$ is of demand type $\D$ for all $j=1 ,\dots,J$. 

\begin{theorem}[Unimodularity Theorem. \cite{BaldwinKlemperer, DKM}] \label{thm:bk}
A set $\D$ of primitive integer vectors is unimodular if and only if every collection of concave valuation functions $\{u^j \mid j = 1,\dots, J\}$ of demand type $\D$ has competitive equilibrium.
\end{theorem}

To prove the above theorem, we first recall a lemma, which has appeared in \cite{BaldwinKlemperer, danilov2004discrete, howard}. For completeness we provide a proof inspired by \cite{BaldwinKlemperer}, which is different from previous proofs.

\begin{lemma} \label{lem:unimodular}
Suppose $\D$ is unimodular. If $P$ and $Q$ are lattice polytopes with edge directions in~$\D$, then $\conv_\Z(P + Q) = \conv_\Z(P) + \conv_\Z(Q).$ 
\end{lemma}

\proof{Proof.}
Translating $P$ and $Q$ to contain the origin and replacing $\Z^n$ with $\Span_\R(P+Q) \cap \Z^n$ if necessary, 
we may assume without loss of generality that $\dim(P+Q) = n$.  
First consider the case when $P$ and $Q$ are contained in complementary affine spaces, that is, when
\begin{equation}\label{eqn:full.dim}
\dim(P+Q) = \dim(P) + \dim(Q).
\end{equation}
Let $B$ and $C$ be maximal linearly independent subsets of primitive edge directions of $P$ and $Q$ respectively.  Then $B \cup C$ is a basis of $\RR^n$ from the dimension assumption above.  Moreover it forms a $\ZZ$-basis for $\Z^n$ because $\D$ is unimodular.  With respect to this basis,  $P$ and $Q$ lie in complementary coordinate subspaces, so the Minkowski sum $P+Q$ coincides with the Cartesian product $P \times Q$, and we have the desired result. 

Now, suppose that (\ref{eqn:full.dim}) does not hold. Let $v \in \RR^n$ be a vector in general position such that the tropical hypersurfaces $T(P)$ and $T(Q) - v$ intersect transversely, i.e.\ two faces intersect only if the dimension of their Minkowski sum is equal to $n$.  Then the union $T(P) \cup (T(Q) - v)$ is dual to a subdivision of $P+Q$ whose faces have the form $\sigma + \tau$ where $\sigma$ and $\tau$ are faces of $P$ and $Q$ with complementary dimensions.  See Figure~\ref{fig:perturb}.
  
Let $a$ be any lattice point in $P+Q$.  Then $a$ belongs to one such face $\sigma+\tau$. Since $\sigma$ are $\tau$ are faces of $P$ and $Q$, their edge directions belong to $\D$ as well.  The same argument as above gives $\sigma+\tau = \sigma \times \tau$ and hence $(\sigma+\tau) \cap \ZZ^n = (\sigma \cap \ZZ^n) + (\tau \cap \ZZ^n)$, which gives $a \in  (P \cap \ZZ^n) + (Q \cap \ZZ^n)$.
\hfill\qed
\endproof

\begin{figure}
\begin{center}
\includegraphics[scale=0.5]{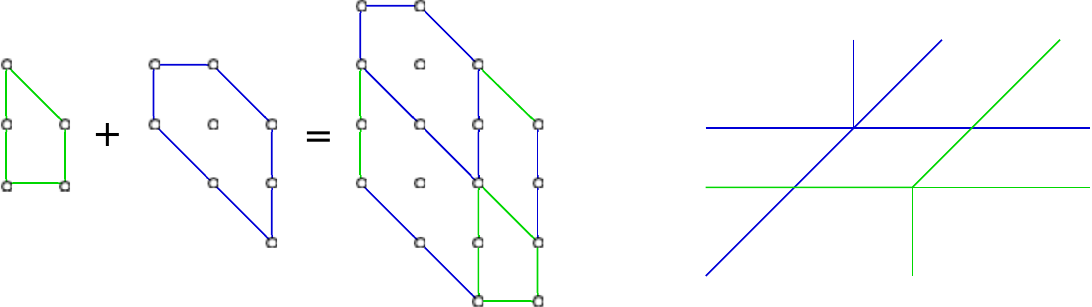}
\end{center}
\caption{The Minkowski sum of two polytopes can be subdivided without introducing new edge directions so that each maximal face is a Minkowski sum of faces of of complementary dimensions.  This can be done, for instance, by translating one of the tropical hypersurfaces in a generic direction and taking the dual subdivision.}
\label{fig:perturb}
\end{figure}

\medskip

\proof{Proof of the Unimodularity Theorem (Theorem \ref{thm:bk}).}
Suppose the set $\D$ is unimodular and that all valuations \mbox{$u^j : A^j \rightarrow \RR$} are concave of demand type $\D$. For each $p \in \R^n$, $\conv(D_{u^j}(p))$ are lattice polytopes with edge directions in $\mathcal{D}$. By Lemma \ref{lem:poly1} and \ref{lem:unimodular}, competitive equilibrium holds for all points in $D_U(p)$. Thus, competitive equilibrium holds by part 4 of Lemma \ref{lem:CEconditions}. For the converse, suppose $\D$ is not unimodular. We need to construct a product-mix auction of demand type $\D$ without competitive equilibrium. We will construct one where the valuations are trivial, that is, zero on their supports and $-\infty$ elsewhere. By Lemma \ref{lem:poly1}, it is sufficient to show that there exists $A^1, \ldots, A^r \subset \Z^n$ with primitive edges in $\D$, $A^j = \conv_\Z(A^j)$, and such that $\sum_jA^j \subsetneq \conv_\Z(\sum_jA^j)$. Since $\D$ is not unimodular, there exist linearly independent vectors $v_1, \ldots, v_n \in \D$ where $\Span_\ZZ\{v_1,\ldots,v_n\} \subsetneq \Z^n$. Let $A^i = \{ \mathbf{0}, v_j \}$.  Then $A^1 + \cdots + A^n$ consists of the vertices of the parallelopiped with edges $v_1, \dots, v_n$, but $\conv_\Z(A^1 + \cdots + A^n)$ consists of other lattice points because of non-unimodularity.
\qed
\endproof

\medskip
\begin{example}[Gross substitutes and $M^\natural$-concave valuations] \label{ex:murota}
In~\cite{KelsoCrawford} Kelso and Crawford introduced the {\em gross substitute} condition for set functions and showed that competitive equilibrium exists under this condition. To explain this property, let us consider a special case. Suppose the set of available bundles for each agent is $\{0,1\}^n$, that is, there is only one item of each type for sale.  A valuation $u$ has the gross substitute property if increasing the price of some items does not decrease the demand of other items. That is, for any price vectors $p \leq q$ and a bundle $a \in D_p(u)$, there is a bundle $b \in D_q(u)$ such that $\{i \in a : p_i = q_i\} \subset b$ where the bundles $a$ and $b$ are considered as subsets of $[n]$.  This implies that the edge directions in $\Delta_u$ cannot have two positive (or two negative) entries.  So the primitive edge directions are in directions $e_i$ or $e_i-e_j$ where $e_i$'s are standard unit vectors in $\ZZ^n$.  These vectors form the unimodular system $A_n^\natural$, which is the projection of the root system $A_n = \{\pm(e_i - e_j) \mid 1 \leq i < j \leq n+1\} \subset \ZZ^{n+1}$ along a coordinate direction. 

In \cite[\S4, \S11]{MurotaBook}, Murota defines $M^\natural$-concave functions, generalizing the gross substitute property to other sets of bundles in $\ZZ^n$. In particular, a concave function $u: A \to \R$, where $A \subset \Z^n$, $A =\conv_\Z(A)$, is $M^\natural$-concave if it can be extended to a concave function on $\conv(A)$, and the faces in the regular subdivision $\Delta_u$ are projections of {\em M-convex sets} along a coordinate direction. The M-convex sets are also known as {\em generalized permutohedra}~\cite{PRW} or {\em polymatroids}. This implies that the set of primitive edge direction of cells in $\Delta_u$ belong to $A_n^\natural$. Since $A_n^\natural$ is unimodular, by the Unimodularity Theorem, competitive equilibrium holds. See \cite[\S11]{MurotaBook} for further discussions on equivalent characterizations of gross substitutes and connections to the min-cost flow problem and algorithms.
\end{example}

\section{Competitive Equilibrium via Integer Programming.}\label{sec:LPIP}
In this section, we show that competitive equilibrium at a point exists if and only if the linear program (\ref{eqn:primal}) below has an integral optimal solution (Theorem \ref{thm:lp.ip.unimod}), or equivalently, that the linear program defined by (\ref{eqn:opt.prob2})-(\ref{eqn:partition}) has an integral optimal solution (Proposition \ref{prop:lp.ip}). 
The first program is new, and it makes explicit the connection to the unimodularity of the demand type $\mathcal{D}$. In particular, via the well-known unimodularity result in integer programming \cite[\S19]{schrijver1998theory}, the Unimodularity Theorem follows as a corollary of Theorem \ref{thm:lp.ip.unimod}. 
The second linear program follows from viewing product-mix auction as a multi-unit combinatorial auction. While it does not immediately yield a proof of the Unimodularity Theorem, it may be computationally and intuitively easier to work with. As an example, we consider the problem of profit maximization under the competitive equilibrium constraint using these linear programs.

\subsection{Integer programming proof of the Unimodularity Theorem}
First we give a sketch of the main ideas. Recall that we lift points in $A \subset \ZZ^n$ to $A \times \RR$ by a height function $U: A \to \R$, then take the concave majorant $\widetilde{U}$, which is a function on $\conv(A)$. We say that a point $a \in A$ is {\em lifted} if $\widetilde{U}(a) = U(a)$.  

Fix a point $a^\ast \in A$. Our goal is to characterize when $a^\ast$ is lifted. Regardless of whether it is lifted or not, $a^\ast$ belongs to some face of the regular subdivision $\Delta_U$ of $A$ induced by $U$. Order cells by set inclusion, say that $\sigma$ is smaller than $\tau$ if $\sigma \subseteq \tau$. Let $D^\ast$ be the smallest cell such that $a^\ast \in \conv(D^\ast)$. Then all vertices of $D^\ast$ are lifted to the same face of the graph of $\widetilde{U}$ in $\conv(A) \times \RR$. 

The point $a^\ast$ can be expressed, not necessarily uniquely, as some reference point $\widetilde{a} \in D^\ast$ plus a rational linear combination of the edge directions of the face $D^\ast$. The aggregate valuation function $U$ is defined to be the maximum among certain sums of individual valuation under decomposing $a^\ast$ into {\em integral} bundles, while $\widetilde{U}(a^\ast)$, which is affine linear on $\conv(D^\ast)$, is the maximum over some {\em rational} linear combinations. Whether $a^\ast$ is lifted can be phrased as whether a particular linear program has an integral solution. 

Furthermore, when the set of edges $\mathcal{D}$ is unimodular, the set of edges of the face $\conv(D^\ast)$ is also unimodular, and we can use this to show that the linear program is guaranteed to have an integral solution. In what follows we make the ideas above precise and also show the choice of reference point $\widetilde{a}$ does not matter, so we have one well-defined integer program and its linear relaxation.  

\medskip
Let $\partial D^\ast$ be the set of vertices of $\conv(D^\ast)$. Let $S(a^\ast)$ denote the set of price vectors that support $\conv(D^\ast)$, that is,
$$
S(a^\ast) = \{p \in \RR^n \mid \widetilde{U}(\widetilde{a}) - p \cdot \widetilde{a} \geq \widetilde{U}(a) - p \cdot a ~~ \text{ for all } \widetilde{a} \in \partial D^\ast, a \in A \}.
$$
For vertex $\widetilde{a} \in \partial D^\ast$, the singleton set $\{\widetilde{a}\}$ is a cell of $\Delta_U$.  Thus $\{\widetilde{a}\} = D_U(p) = D_{u^1}(p) + \cdots + D_{u^J}(p)$ for some $p \in \RR^n$, so $\widetilde{a}$ can be written uniquely as $\widetilde{a} = \sum_j \widetilde{a}^j$ where each $\widetilde{a}^j$ is a vertex of the subdivision $\Delta_{u^j}$ of $A^j$. Then $S(a^\ast)$ is the set of $p \in \R^n$ such that
\begin{equation}
\label{eqn:ineqs}
u^j(b^j) - p \cdot b^j \leq u^j(\widetilde{a}^j) - p \cdot \widetilde{a}^j \text{ for all } \widetilde{a} \in \partial D^\ast,  b^j \in A^j, \text{ for all } j = 1, \ldots, J.
\end{equation}
After rewriting, we get that 
$$ p \cdot (\widetilde{a}^j - b^j) \leq u^j(\widetilde{a}^j) - u^j(b^j)  \text{ for all } b^j \in A^j, j=1,\dots,J \mbox{ and for all } \widetilde{a} \in \partial D^\ast. $$
Suppose there is a total of $N$ such constraints on $p$. Let $V$ be the $n \times N$ matrix whose columns are the vectors $\widetilde{a}^j - b^j$ as above, taken over all $\widetilde{a}^j$ and $b^j$'s. Let $c \in \R^N$ be the vector with entries 
\begin{equation}
\label{eq:def.c} u^j(\widetilde{a}^j) - u^j(b^j)\end{equation}
corresponding to columns $\widetilde{a}^j - b^j$ of $V$ respectively. Then we can rewrite $S(a^\ast)$ as
$$ S(a^\ast) = \{p \in \RR^n \mid V^\top p \leq c\}. $$
Let $\widetilde a$ be any point in $\partial D^\ast$. Consider the following program in decision variable $p \in \R^n$:
\begin{align}
\text{maximize~~~~}& p \cdot (\widetilde{a} - a^*) & \label{eqn:dual} \tag{D} \\
\text{subject to~~~~}& V^\top p \leq c. \nonumber
\end{align}
It is the dual of the following primal linear program in decision variable $x \in \R^N$
\begin{align}
\text{minimize~~~~}& c^\top x & \label{eqn:primal} \tag{P} \\
\text{subject to~~~~}& Vx = \widetilde{a} - a^*, ~~ x \geq 0. \nonumber
\end{align}

\begin{theorem}\label{thm:lp.ip.unimod} Let $a^* \in A$ and $\widetilde a \in \partial D^\ast$. 
Competitive equilibrium for $\{u^j\}$ exists at $a^\ast$ if and only if the optimum of the linear program (\ref{eqn:primal}) over $x \in \R^N$ equals its optimum over $x \in \Z^N$.
\end{theorem}

We break the proof of Theorem \ref{thm:lp.ip.unimod} into two smaller lemmas. Lemma \ref{lem:prop8.1} says that the optimal solutions of the programs (\ref{eqn:dual}) and (\ref{eqn:primal}) are independent of the choice of $\widetilde{a}$. Lemma \ref{lem:prop8.2} says that the objective of (\ref{eqn:primal}) is minimized exactly when $a^\ast$ is lifted. 
\begin{lemma}\label{lem:prop8.1}
Every feasible solution is optimal in (\ref{eqn:dual}), with objective function value $U(\widetilde{a}) - \widetilde{U}(a^\ast)$.
\end{lemma}

\medskip
\proof{Proof.}
Consider the graph of the function $\widetilde{U}$ on $\conv(A)$, which is a subset of $\conv(A) \times \RR$.  All points in the face $\conv(D^\ast)$ is lifted accordingly to an upper face the convex hull of the graph. For any $p \in S(a^\ast)$, the vector $(-p,1)$ supports this face. Since both $a^*$ and $\widetilde a$ belong to $\conv(D^\ast)$, we have 
$$\widetilde{U}(a^*) - p \cdot a^* =\widetilde{U}(\widetilde a) - p \cdot \widetilde{a}.$$
Then the value of the objective function in (\ref{eqn:dual}) is 
\begin{equation}
\label{eq:optValue}
p \cdot (\widetilde{a} - a^*) =  \widetilde{U}(\widetilde{a}) - \widetilde{U}(a^\ast) =  U(\widetilde{a}) - \widetilde{U}(a^\ast).\end{equation}
The second equality follows since $\widetilde{a}$ is a vertex, hence a marked point, in the subdivision $\Delta_U$.  Since the value $U(\widetilde{a}) - \widetilde{U}(a^\ast)$ does not depend on $p$, it follows that every feasible solution is optimal in (\ref{eqn:dual}).
\hfill\qed
\endproof

\begin{lemma}\label{lem:prop8.2}
Any feasible solution $x \in \Z^N$ of (\ref{eqn:primal}) satisfies 
\begin{equation}\label{eqn:objective.primal.x}
c^\top x \geq U(\widetilde{a}) - \widetilde{U}(a^*).
\end{equation}
Equality holds for some feasible $x \in \Z^N$ if and only if $U(a^*) = \widetilde{U}(a^*)$. 
\end{lemma}

\medskip
\proof{Proof.}
By Lemma \ref{lem:prop8.1}, we know that $U(\widetilde{a}) - \widetilde{U}(a^\ast)$ is equal to the optimal objective value $p \cdot (\widetilde{a} - a^\ast)$ of the primal problem (\ref{eqn:primal}). Weak linear programming duality \cite[\S 7]{schrijver1998theory} then implies (\ref{eqn:objective.primal.x}). 

It remains to show the condition for equality to occurs. Any feasible $x$ satisfies $\widetilde{a} - Vx= a^\ast$. Recall that columns of $V$ are vectors pointing from points in $A^j$ toward vertices of $\Delta_{u^j}$. Since $\widetilde{a}$ is a vertex of $D^\ast$, there exists a unique decomposition
$$\widetilde{a}=\widetilde{a}^1 + \dots + \widetilde{a}^J$$ 
such that $\widetilde{a}^j$ is a vertex of $D_{u^j}(p)$, for $p$ such that $D^\ast = D_U(p)$.  Hence any feasible solution $x \in \ZZ^n$ gives a way of writing
\begin{equation}\label{eqn:decomp.a}
a^* = (\widetilde{a}^1 + v^1) + \dots + (\widetilde{a}^J + v^J),
\end{equation} 
for some vectors $v^1, \ldots, v^J$ such that each $\widetilde{a}^j + v^j$ lies in the integer affine span of $A^j$.

Consider the lift of $\widetilde{a}^j + v^j$ to $\RR^n \times (\RR \cup \{-\infty\})$ by the valuation $u^j$.  If  $\widetilde{a}^j + v^j \notin A^j$, then it is lifted to $- \infty$. Since $u^j$ is concave and $(-p,1)$ supports the lift of $D^\ast$, we have
\begin{equation}\label{eqn:inequal.vj}
u^j(\widetilde{a}^j + v^j) - p \cdot (\widetilde{a}^j + v^j) \leq u^j(\widetilde{a}^j) - p \cdot \widetilde{a}^j
\end{equation}
with equality if and only if $\widetilde{a}^j + v^j \in D_{u^j}(p)$.

Summing over all $j$'s, we get
\begin{equation}
\label{eq:sum}
\sum_{j=1}^J u^j(\widetilde{a}^j + v^j) - p \cdot a^* \leq U(\widetilde{a}) - p \cdot \widetilde{a}.
\end{equation}
Hence
\begin{equation}
\label{eq:objective.primal.x.proof}
 U(\widetilde{a}) - \widetilde{U}(a^*) \stackrel{\eqref{eq:optValue}}{=} p \cdot (\widetilde{a} - a^*) \stackrel{\eqref{eq:sum}}{\leq} U(\widetilde{a}) - \sum_{j=1}^J u^j(\widetilde{a}^j + v^j) = \sum_{j=1}^J\left(u^j(\widetilde{a}^j) - u^j(\widetilde{a}^j + v^j)\right) = c^\top x,
\end{equation}
where $x$ gives the coefficients in writing the $v^j$'s as sums of columns in $V$; see \eqref{eqn:decomp.a}.  The last equality in \eqref{eq:objective.primal.x.proof} then follows from the definition of the vector $c$ in \eqref{eq:def.c}.

Equality is achieved in \eqref{eqn:objective.primal.x} and \eqref{eq:objective.primal.x.proof} for some feasible $x \in \Z^N$ if and only if equality is achieved in (\ref{eqn:inequal.vj}) for every $j = 1, \ldots, J$. This happens if and only if the decomposition of $a^*$ in (\ref{eqn:decomp.a}) by such $x$ satisfies $\widetilde{a}^j + v^j \in D_{u^j}(p)$ for some $p \in S(a^*)$, for all $j = 1, \ldots, J$. But this is the definition of $a^\ast \in D_U(p)$. So, assuming that (\ref{eqn:primal}) is feasible, equality occurs in (\ref{eqn:objective.primal.x}) if and only if $U(a^\ast) = \widetilde{U}(a^\ast)$. 
\hfill\qed
\endproof

\medskip
\proof{Proof of Theorem \ref{thm:lp.ip.unimod}.}
Consider Lemma~\ref{lem:prop8.1}. By the linear programming duality, the optimal value of the linear program (\ref{eqn:primal}) over $x \in \R^N$ is also $U(\widetilde{a}) - \widetilde{U}(a^\ast)$. If (\ref{eqn:primal}) is feasible over $x \in \Z^N$, then Lemma~\ref{lem:prop8.2} implies the desired statement. 
If (\ref{eqn:primal}) is not feasible over $x \in \Z^N$, then $a^*$ is not in $A$, so there is no competitive equilibrium at $a^\ast$, and the linear and integer program (\ref{eqn:primal}) do not agree. Thus the desired statement also holds true in this case.
\hfill\qed
\endproof

\begin{corollary}
Theorem \ref{thm:lp.ip.unimod} implies the Unimodularity Theorem. 
\end{corollary}
\proof{Proof. } Suppose that all the valuation $u^j$ are concave and of unimodular demand type.  Then any point $b^j \in A^j$ can be reached from any vertex $\widetilde{a}^j$ of the subdivision $\Delta_{u^j}$ by walking along the primitive edge directions of the subdivision because of the unimodularity of the set of directions.  Thus in (\ref{eqn:dual}) and (\ref{eqn:primal}), we can remove from $V$ and $c$ all the columns except those corresponding to the primitive edge directions of the subdivisions.  We are then left with a unimodular matrix $V$ which gives the same optima for the linear and integer programs as before.
A well known result in integer programming states that (\ref{eqn:primal}) always has integral optimal solutions (see \cite[\S19]{schrijver1998theory}), so by Theorem~\ref{thm:lp.ip.unimod} competitive equilibrium exists for all points $a^\ast \in A$.  This is precisely the difficult direction of the Unimodularity Theorem. 
\hfill\qed
\endproof
\medskip

\subsection{Weighted set packing and the Cayley trick.}\label{sec:lpip.cayley}

Let us fix a supply bundle $a^\ast$ and assume that $\mathbf 0 \in A^j$ for every $j$.   Consider the following set packing problem where the decision variables are $y(a,j)$ for $j = 1,\dots,J$ and $a \in A^j$.
\begin{align}
\mbox{maximize~~~~} & \sum_{j=1}^J \sum_{a \in A^j} y(a,j) u^j(a) \label{eqn:opt.prob2}  \\
\mbox{subject to~~~~}& y(a,j) \geq 0 \text{ for all }j=1,\dots,J \text{ and } a \in A^j \label{eqn:assign.something} \\
 & \sum_{a \in A^j} y(a,j) = 1 \text{ for all }j=1,\dots,J \label{eqn:assign.one}  \\
 & \sum_{j=1}^J \sum_{a \in A^j} y(a,j) a = a^\ast \label{eqn:partition}
\end{align}
For a bundle $a \in \ZZ^n$,  $y(a,j) = 1$ means that we assign the bundle $a$ to agent $j$. The constraints (\ref{eqn:assign.something}), (\ref{eqn:assign.one}) and (\ref{eqn:partition}) state, respectively, that each agent $j$ is assigned a non-negative fraction of bundles, that the total mass of bundles they receive is one, and that a total of $a^\ast$ items is assigned to all agents together. The objective (\ref{eqn:opt.prob2}) asks for an assignment that maximizes the total valuation. 

\begin{proposition}\label{prop:lp.ip}
Competitive equilibrium for $\{u^j\}$ exists at $a^\ast$ if and only if there exists an integral optional solution to the linear program defined by (\ref{eqn:opt.prob2})-(\ref{eqn:partition}), that is, an optimal solution for which all $y(a,j)$ are integers.
\end{proposition}

\medskip
\proof{Proof.}
By definition (\ref{eqn:U}), the aggregated valuation $U(a^\ast)$ is the optimum for the \emph{integer program}~(\ref{eqn:opt.prob2}). Let $\widetilde{U}$ be the concave majorant of $U$ and $\widetilde{u}^j$ be the concave majorant of $u^j$ for $j = 1, \ldots, J$. Note that 
\begin{equation}\label{eqn:tilde.U}
\widetilde{U}(a^\ast) = \max\left\{\sum_{j=1}^J \widetilde{u}^j(a^j): a^j \in \conv(A^j)\text{ and }\ \sum_{j \in J} a^j = a^\ast \right\}. 
\end{equation}
Since each $u^j$ is concave, $u^j = \widetilde{u}^j$ on each $A^j$. Since $\conv(A) = \conv(A^1) + \ldots + \conv(A^J)$, for any \mbox{$a^\ast \in \conv(A) \cap \Z^n$} the concave majorant $\widetilde{U}(a^\ast)$ is the optimum for the \emph{linear program}~(\ref{eqn:opt.prob2}). Existence of competitive equilibrium at $a^\ast$ means that $U(a^\ast) = \widetilde{U}(a^\ast)$, so we have the statement of the proposition.\hfill\qed
\endproof

The linear program (\ref{eqn:opt.prob2})-(\ref{eqn:partition}) can be rephrased in terms of the ``Cayley trick'' as follows.  The {\em Cayley configuration} of the collection $A^1,\dots, A^J \subset \ZZ^n$ is a point configuration in $\ZZ^J \times \ZZ^n$ consisting of the points $$\Cay(A^1, \dots, A^J) = (\{e_1\} \times A^1) \cup \cdots \cup (\{e_J\} \times A^J).$$  
The Cayley trick says that there is a natural bijection between mixed subdivisions of the Minkowski sum $A = \sum_iA_i$ and subdivisions of the Cayley configuration $\Cay(A^1, \dots, A^J)$ \cite{sturmfels1994newton}. There are powerful theories and software for understanding the later \cite{jesus2010loera}, and thus this trick is often used to compute and understand mixed subdivisions. Let $C$ be the matrix whose columns are the points in $\Cay(A^1,\dots,A^J)$. Let $u$ be the vector of valuations $u^j(a)$ whose entries naturally correspond to columns of $C$.  Then  (\ref{eqn:opt.prob2})--(\ref{eqn:partition}) can be stated simply as
\begin{align}
\label{eqn:Cayley}
\text{maximize~~} & u \cdot y \\
\text{subject to~~} & y \geq 0 \text{ and } C y = \begin{pmatrix}\mathbf{1}\\a^*\end{pmatrix}
\end{align}
where $\mathbf{1}$ is the all-one vector in $\ZZ^J$. That is, competitive equilibrium is a special optimization question defined using Cayley configurations, which coincides with the weighted set packing problem.

The two linear programs in Theorem \ref{thm:lp.ip.unimod} and Proposition \ref{prop:lp.ip} are different ways to view competitive equilibrium. Depending on the problem, one formulation can be easier to work with. To illustrate, suppose competitive equilibrium exists at $a^\ast$. Suppose we want to find a price $p$ that gives competitive equilibrium at $a^\ast$ and also maximizes the profit $p\cdot a$ for the seller at the same time.

From the viewpoint of the programs (\ref{eqn:dual}) and (\ref{eqn:primal}), one needs to solve
\begin{align}
\text{maximize~~~~}& p \cdot a^\ast & \label{eqn:optim.price} \\
\text{subject to~~~~}& V^\top p \leq c.  \nonumber 
\end{align}
However, forming $V$ and $c$ requires one to find all the vertices $\widetilde{a}$ of $D^\ast$ and their decomposition into sums of vertices $\widetilde{a}^j$ in $\Delta_{u^j}$. In this case, the set packing view (\ref{eqn:opt.prob2})-(\ref{eqn:partition}) gives a more efficient formulation below. Similar questions were studied by Bikhchandani and Mamer in~\cite{BM97}. 

\begin{lemma}
Let $\left(y^\ast(a,j)\right)_{j=1,\dots,J; a \in A^j}$ be an optimal integral solution of (\ref{eqn:opt.prob2}). The equilibrium price $p \in \RR^n$ that maximizes profit for the seller is the solution to the following linear program
\begin{align}
\mbox{\emph{maximize}~~~~}& p \cdot a^\ast & \label{eqn:max.profit.lp} \\
\mbox{\emph{subject to}~~~~}& u^j(a) - p \cdot a \geq u^j(b) - p \cdot b, \nonumber \text{for every } j = 1, \ldots, J, \\
&\mbox{ every } a \in A^j \mbox{ with } y^\ast(a,j) = 1, \mbox{ and every } b \in A^j \mbox{ adjacent to } a \mbox{ in } \Delta_{u^j}. \nonumber
\end{align}
\end{lemma}

\medskip
\proof{Proof.}
Let $F$ be the set of feasible solutions of the linear program in the Lemma. We wish to show that $F = S(a^\ast) = \{p \in \RR^n \mid a^\ast \in D_U(p)\}$.  

If $p \in F$, then $\{a \in A^j \mid y^\ast(a,j) = 1\} \subset D_{u^j}(p)$, so $a^\ast \in D_U(p)$, since
$$a^\ast = \left(\sum_{a^j \in A_j : y^\ast(a^j,j) = 1}a^j\right) ~\in~ \sum_{j=1}^J D_{u^j}(p) = D_U(p).$$  This shows that $F \subseteq S(a^\ast)$.  

Now suppose that $p \in S(a^\ast)$. That is, there are $b^j \in D_{u^j}(p)$ such that $a^\ast = \sum_{j=1}^J b^j$.  For each $j$, let $a^j$ be the unique element of $A^j$ for which $y^\ast(a^j,j)=1$.  Since $b^j \in D_{u^j}(p)$, we have $u^j(b^j) - p \cdot b^j \geq u^j(a^j) - p \cdot a^j$, and
$$
U(a^\ast) - p \cdot a^\ast = \sum_{j=1}^J \left( u^j(b^j) - p \cdot b^j \right) \geq  \sum_{j=1}^J \left( u^j(a^j) - p \cdot a^j  \right) = U(a^\ast) - p\cdot a^\ast. 
$$
This shows that all the inequalities in~(\ref{eqn:ineqs}) are attained at equality, so $u^j(b^j) - p \cdot b^j = u^j(a^j) - p \cdot a^j$ for all $j$.  Thus $a^j \in D_{u^j}(p)$ for all $j$, and $p$ is a feasible solution of (\ref{eqn:max.profit.lp}).
\hfill\qed
\endproof

\subsection{Competitive equilibrium and subset sum.}\label{sec:polytime}
Our various formulations of competitive equilibrium result in several algorithms on verifying the existence of competitive equilibrium at a given point $a^\ast$. Without further assumptions, they all involve the subset sum problem, and are thus all NP-complete. 

There are different layers of difficulties. Firstly, for an arbitrary bundle $a^\ast \in \ZZ^n$, verifying if \mbox{$a^\ast \in A = \sum_{j=1}^J A^j$} is a subset sum problem, even when there is only one type of good ($n = 1$).  
Secondly, suppose we know that $a^\ast \in A$. One may assume as in Proposition \ref{prop:lp.ip} that $u^j$ is concave and $\mathbf{0} \in A^j$ for every $j$. Then for $n = 1$, all agents have unimodular demand type, so competitive equilibrium trivially holds for $n = 1$. For $n > 1$, however, the Unimodularity Theorem may not apply. By Lemma \ref{lem:CEconditions}, competitive equilibrium at $a^\ast$ holds if and only if $U(a^\ast) = \tilde{U}(a^\ast)$ where $U$ is the aggregate valuation and $\tilde{U}$ is its concave majorant.
 Computing $U(a^\ast)$ is again a subset sum problem. This is spelled out in Proposition \ref{prop:lp.ip}. Knowing that $a^\ast \in A$ means the linear program \eqref{eqn:opt.prob2} has a feasible real solution over $\R$, but this does not guarantee an integral solution. 

Alternatively, suppose we know one price $p \in \R^n$ such that $a^\ast \in D_U(p)$. 
As the individual valuations $\{u^j\}$ are concave, one can compute $D_{u^j}(p)$. One can then setup the integer program \eqref{eqn:primal}, whose solution guarantees competitive equilibrium by Theorem \ref{thm:lp.ip.unimod}. The program \eqref{eqn:primal} asks whether one can write $a^\ast$ as $\tilde{a}$ plus an integral combination from a given set of vectors $V$. It is yet another instance of subset sum. 

Under the assumption that $a^\ast \in A$, the Unimodularity Theorem replaces the subset sum problems in \eqref{eqn:opt.prob2} and \eqref{eqn:primal} by a sufficient, easy-to-check condition (namely, unimodularity of a set of vectors) for competitive equilibrium to hold at $a^\ast$. 

The Subgroup Indices Theorem of Baldwin and Klemperer \cite[Theorem 5.16]{BaldwinKlemperer} give an alternative criterion when the tropical hypersurfaces $\{T_{u^j}\}$ have transverse intersections. This means the values $u^j(a)$ for $a \in A^j$ are sufficiently general, so that all intersections of the tropical hypersurfaces $\mathcal{T}(f_{u^j})$ locally looks like transverse intersection of affine spaces. This can always be achieved by adding a small real number to each value $u^j(a)$, for example. Assuming that one knows the face $D_U(p)$ and its decomposition as Minkowski sums of $D_{u^j}(p)$ for $j=1,\ldots,J$, transverse intersection implies that all but at most $n$ of the faces $D_{u^1}(p), \ldots, D_{u^J}(p)$ are vertices, say, $D_{u^j}(p) = \{a^j\}$ for $j = n+1, \ldots, J$. This effectively reduces the dimension of the subset sum problem, from decomposing $a^\ast$ into a sum of $J$ terms to that of at most $n$ terms. See \cite{BaldwinKlemperer} and discussions therein on sufficient conditions for competitive equilibrium in this setting. From a computational viewpoint, transverse intersections alone do not guarantee competitive equilibrium at some given $a^\ast$. Instead, it guarantees that for fixed $n$, competitive equilibrium can be jointly checked at all points in time polynomial in $J$. 

\section{Stable auctions and the Oda conjecture}\label{sec:oda}
\subsection{Stable auctions and competitive equilibrium for $n=2$}
Consider a product-mix auction with $J$ agents and $n$ product types. Let us partition the agents into disjoint, non-empty subsets, and compute the product-mix auction on each one. What properties of these subset auctions would guarantee that the original auction has competitive equilibrium? Such divide-and-conquer conditions are naturally attractive, both theoretically and computationally. In this section, we give such a sufficient condition when $n=2$ (Theorem \ref{thm:n2}). The analogous statement in higher dimensions is equivalent to the Oda conjecture in toric geometry, discussed in Section \ref{sec:conjecture}. 

Suppose all the subset auctions individually have competitive equilibrium. This condition is clearly not sufficient to guarantee competitive equilibrium overall. For instance, when each subset contains exactly one agent, the existence of competitive equilibrium in each subset is equivalent to individual valuations being concave. However, there are product-mix auctions with concave individual valuations but without competitive equilibrium. Counter-intuitively, we show that individual competitive equilibrium in the subset auctions is not even necessary. That is, our conditions in Theorem \ref{thm:n2} allow for competitive equilibrium to \emph{fail} for some of the subset auctions, yet they guarantee that the joint auction has competitive equilibrium. 

Theorem \ref{thm:n2} is complementary to results in \cite{BaldwinKlemperer}, which focus on transverse intersections of tropical hypersurfaces. In a way, our setup is as far as possible from transverse intersections: the hypothesis implies that the tropical hypersurface for a subset of the agents contains the union of the rest. 

\begin{definition}[Stable auction]\label{defn:stable}
Consider a product-mix auction with $J \geq 1$ agents with concave valuations, and $n \geq 1$ product types. Partition the agents into $K \leq J$ disjoint, non-empty sets, and run the product-mix auction separately on each set. Let $U^k$ be the aggregate valuation  of the $k$-th set of agents, for $1 \leq k \leq K$. We call such a partition \emph{stable} if the first subset auction has competitive equilibrium, and
\begin{equation}\label{eqn:t.subset}
T(f_{U^k}) \subseteq T(f_{U^1}) \mbox{ for all } k = 1, \ldots, K.
\end{equation}
We say that a product-mix auction is \emph{stable} if it has a stable partition.
\qed
\end{definition}

Intuitively, the stable condition (\ref{eqn:t.subset}) says that the first subset auction is `rich' enough to capture much information about the original auction. Note that we require competitive equilibrium to hold only for the first subset auction. For the other subset auctions, competitive equilibrium is allowed to fail. Instead, we put restrictions on their tropical hypersurfaces via (\ref{eqn:t.subset}). There are two ways to interpret this condition. First, by Lemma \ref{lem:simple.operations}, 
$$T(f_U) = \bigcup_{j=1}^J T(f_{u^j}) = \bigcup_{k=1}^K T(f_{U^k}) = T(f(U^1)).$$
In other words, suppose one were to iteratively compute $T(f_U)$, the set of price indifferences, by adding in one agent at a time. Then (\ref{eqn:t.subset}) means that this set does not change after the first subset of agents are processed. That is, the locus of indifference prices stabilizes, hence the name stable auction. This notion of stability is  about union of tropical hypersurfaces, not to be confused with stable intersections in tropical geometry. 

One can also understand (\ref{eqn:t.subset}) in terms of the demand sets $D_{U^1}, \ldots, D_{U^K}$. For each price $p \in \R^2$, we can count how many product  bundles are demanded in each subset auction at this price, ignoring scalings by constant multiples. Lemma \ref{lem:d.supset} states that if for all prices $p$, this number does not go up after the first subset auction, then (\ref{eqn:t.subset}) holds.  

We now give a simple sufficient condition for the condition (\ref{eqn:t.subset}). The following definition generalizes the concept of a primitive edge direction.
\begin{definition}
For $P \subset \ZZ^n$, the primitive polytope of $P$, denoted $P_\downarrow$, is the smallest polytope $Q \subset \Z^n$ such that $c \cdot Q = P$ for some $c \in \mathbb{N}$.
\end{definition}

\begin{lemma}
\label{lem:d.supset}
The statement (\ref{eqn:t.subset}) is true if
\begin{equation}\label{eqn:d.supset}
|(D_{U^1}(p))_\downarrow| \geq |(D_{U^k}(p))_\downarrow| \mbox{ for all } k = 1,\ldots,K.
\end{equation}
That is, at any price $p \in \R^2$, there are at least as many primitive demanded good bundles in the first auction as those in the others. 
\end{lemma}
\proof{Proof.}
Note that for any polytope $P \subset \Z^n$, $|P| \geq |P_\downarrow| \geq 1$, and $|P_\downarrow| = 1$ if and only if $|P| = 1$. If $-p \in T(f_{U^k})$, then $|D_{U^k}(p)| \geq 2$, so $|(D_{U^k}(p))_\downarrow| \geq 2$. By (\ref{eqn:d.supset}), $|D_{U^1}(p)| \geq 2$. Thus $-p \in T(f_{U^1})$, so (\ref{eqn:t.subset}) holds.
\hfill \qed
\endproof

\begin{theorem}\label{thm:n2}
For $n = 2$ product types, a stable product-mix auction has competitive equilibrium. 
\end{theorem}
\proof{Proof.}
Consider a stable partition with $K$ subset auctions.
By Lemma~\ref{lem:simple.operations}, $T(f_U) = \bigcup_{k=1}^K T(f_{U^k}) = T(f_{U^1})$. Therefore, without loss of generality, one can fix the first auction, and assume that all the remaining auctions consist of a single agent only. Now, let $p$ be a vertex of $-T(f_U)$. By Lemma \ref{lem:CEconditions}, it is sufficient to show that competitive equilibrium holds for all points in $\conv_\Z D_U(p)$. Since each valuation $U^k$ is concave by assumption, by Lemma \ref{lem:poly1}, we need to show that
\begin{equation}\label{eqn:need.hnps}
\conv_\Z(D_{U^1}(p) + \ldots + D_{U^K}(p)) = \conv_\Z(D_{U^1}(p)) + \ldots + \conv_\Z(D_{U^K}(p)).
\end{equation}
As $T(f_U) = T(f_{U^1}) \supseteq T(f_{U^k})$, the normal fan of $D_{U^1}(p)$ refines the normal fan of $D_{U^k}(p)$ for all $k=1,\ldots,K$. Theorem 1 of \cite{HNPS} precisely states that this normal fan condition implies (\ref{eqn:need.hnps}). This concludes the proof.
\hfill \qed
\endproof

\begin{corollary}\label{cor:n2.stable}
For $n=2$ product types, if all agents have identical concave valuations, then competitive equilibrium exists for all relevant supply bundles, regardless of demand types.
\end{corollary}

The next example shows that a product-mix auction can be stable even though some subset of agents together do not have a competitive equilibrium.

\begin{example}\label{ex:n2.oda}
Suppose $n=2$ and $J=3$. Let $A^1 = \{(0,0), (1,2), (1,0), (1,1),(2,2)\}$, $A^2 = \{(0,0), (1,0)\}$, $A^3 = \{(0,0), (1,2)\}$. Let the valuations $u^1, u^2, u^3$ be zero on $A^1, A^2$ and $A^3$, respectively, and $-\infty$ elsewhere. Partition the agents into $K=2$ sets, the first consists of agent 1, the second consists of agents 2 and 3. Here $\mathcal{D} = \{(1,0), (1,2)\}$ is not a unimodular set, and indeed, competitive equilibrium fails at $(1,1)$ for the second subset auction. The first auction with just agent 1 has competitive equilibrium. One can verify that (\ref{eqn:d.supset}) holds. By Theorem \ref{thm:n2}, the auction with all three agents has competitive equilibrium. This is readily verified in Figure \ref{fig:n2.oda}. 
\end{example}

\begin{figure}
\begin{center}
\includegraphics[scale=0.8]{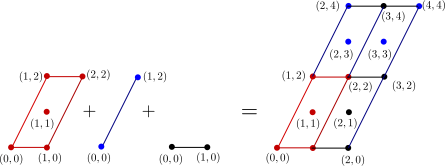}
\end{center}
\caption{The regular subdivisions of the three individual valuations (left) and the aggregate valuation (right). Figure accompanies Example~\ref{ex:n2.oda}, where a stable product-mix auction has competitive equilibrium for all but fails for a subset of the agents.}
\label{fig:n2.oda}
\end{figure}

\subsection{Identical valuations and the integer decomposition property}

The analogue of Corollary \ref{cor:n2.stable} fails for three or more product types.
That is, for $n \geq 3$, there exists product-mix auctions where all agents have identical concave valuations, yet competitive equilibrium fails. The immediate consequence is that the analogue of Theorem \ref{thm:n2} fails in dimension $\geq 3$. 
\begin{proposition}\label{ex:white.43}
For $n \geq 3$, there are auctions where
\begin{itemize}
  \item all the agents have identical concave valuations, and
  \item competitive equilibrium fails.
\end{itemize} 
In particular, for $n \geq 3$, there are stable auctions without competitive equilibrium.
\end{proposition}
\proof{Proof.} 
It is sufficient to do one explicit example for $n=3$.
Consider two agents $(J = 2)$ with three product types $(n = 3)$. Let
$$ A' := \{(0,0,0), (1,1,0), (1,0,1), (0,1,1)\}. $$
Note that $A' = \conv_\Z A'$. 
Let $u': \Z^3 \to \R$ be any function supported on $A'$, and $-\infty$ elsewhere. 
Let the two agents have identical valuations $u^1 = u^2 = u'$. As points in $A'$ are affinely independent, the corresponding subdivision $\Delta_{u'}$ is trivial. Let $A = A' + A' = \{a+b: a \in A', b \in A'\}$. 
The vertices of $A$ are the four points 
$$\{(0,0,0), (2,2,0), (2,0,2), (0,2,2)\}. $$
Let $a^\ast = (1,1,1)$. Then $a^\ast = \frac{1}{4}(0,0,0) + \frac{1}{4}(2,2,0) + \frac{1}{4}(2,0,2) + \frac{1}{4}(0,2,2)$, so $a^\ast \in \conv_{\ZZ}(A)$. However, $a^\ast \notin A$, so competitive equilibrium fails at~$a^\ast$.
\qed
\endproof \medskip

Let us consider what the Unimodularity Theorem has to say about the counterexample in the proof of Proposition \ref{ex:white.43}. Let $\mathcal{D}$ be the set of primitive edge directions of $A'$. This set is not unimodular: among other vectors, it contains $(0,1,1), (1,0,1), (1,1,0)$, whose span over $\Z$ does not contain $(1,1,1)$, for instance. The Unimodularity Theorem says that from $\mathcal{D}$, there exists \emph{some} product-mix auction of demand type $\mathcal{D}$ that fails competitive equilibrium. We exhibited such an example in its proof. The theorem does not imply that the particular case in the proof of Proposition \ref{ex:white.43} has to fail. 

The example above is an instance of a general question about integral polytopes. Suppose we fix an integral polytope $P \subset \RR^n$, and suppose all the $J$ agents have identical valuations, which are $0$ on $P \cap \ZZ^n$ and $-\infty$ else. The trivial partition with one agent in each subset is stable. Thus, our auction is stable. By Lemma~\ref{lem:poly1}, competitive equilibrium means
\begin{equation}\label{eqn:normal.all.equal}
(J \cdot P) \cap \ZZ^n = \underbrace{(P \cap \ZZ^n) + \cdots + (P\cap \ZZ^n)}_{J \text{ times }}
\end{equation}
A polytope $P$ that satisfies (\ref{eqn:normal.all.equal}) for all  $J \in \mathbb{N}$ is said to have \emph{IDP} (integer decomposition property). Polytopes with IDP are also called \emph{integrally closed} or \emph{normal}, but we will use the term IDP because it has the most consistent meaning across literature.  Corollary \ref{cor:n2.stable} says that for $n \leq 2$, all integral polytopes have IDP~\cite{HNPS}. This is no longer true for $n \geq 3$. The tetrahedron $A'$ in Example~\ref{ex:white.43} is an integral polytope that does not have IDP. However, a result of Bruns, Gubeladze, and Trung~\cite[Theorem~1.3.3]{BGT} says that for any integral polytope $P$ and any integer $c \geq \dim(P)-1$, the dilated polytope $cP$ has IDP.  See \cite[Theorem~1.1]{CHHH} for an easy proof. 
We can interpret this result in terms of product mix auctions as follows.


\begin{definition}[$c$-refinement]
For $c \in \mathbb{N}$, the $c$-refinement of a valuation $u: A \subset \Z^n \to \R$ is a function $\bar{u}: c \cdot A \subset \Z^n \to \R$, such that $\bar{u}$ is the concave majorant of function defined by $\bar{u}(c \cdot a) = c \cdot u(a)$ for all $a \in A$. The $c$-refinement of a product-mix auction with valuations $(u^1, \ldots, u^J)$ is one with valuations $(\bar{u}^1, \ldots, \bar{u}^J)$, where $\bar{u}^j$ is the $c$-refinement of $u^j$ for $j=1,\ldots,J$.
\end{definition}

\smallskip

Refinements reflects the effect of changing the base unit of in a product-mix auction. For instance, suppose that instead of selling bananas and apples in boxes of 10 each, one now sells them individually. Instead of querying the agents at the new points, one infers their values by discretizing the original valuation along a grid with width $1/10$, and then multiply by 10 to get back a product-mix auction. Doing this for all agents result in a product-mix auction in a finer unit, hence the name refinement. 

The $c$-refinement of a valuation function $u$ does not change the set of primitive edges in $\Delta_u$, and thus does not change its demand type. The $c$-refinements of an arbitrary auction can fail to have competitive equilibrium. The Bruns--Gubeladze--Trung theorem translates to a special class of stable auctions where there always exists $c$-refinements with competitive equilibrium.  This can be seen as a weak analogue of Corollary \ref{cor:n2.stable} in higher dimension. 

\begin{theorem}[\cite{BGT, CHHH}]
Consider an auction with $n$ product types, where all agents have identical concave valuations $u: A \to \R$ for some $A \subset \Z^n$. Then the $c$-refinement of this auction has competitive equilibrium for any $c \geq n-1$.  
\end{theorem}
In other words, in this situation competitive equilibrium exists if we are willing to subdivide each unit into $n-1$ pieces.  The result holds regardless of the demand types.

\smallskip

\proof{Proof.}
Let $\Delta_u$ denote the regular subdivision of $A$, the demand complex, for each agent. For $J$ agents, with the aggregate valuation $U$, each cell of the regular subdivision $\Delta_U$ has the form $\sigma + \cdots + \sigma$ ($J$~times) for some cell $\sigma$ of $\Delta_u$. Since $u$ is concave, $\conv_\Z(\sigma) = \sigma$. 
Then by Bruns--Gubeladze--Trung, for any integer $c \geq n-1$, the dilation $c \cdot \sigma$ has IDP for all cells $\sigma$ in $\Delta_U$.  This implies that
$J \cdot (c \cdot \sigma) = \conv_\Z(J \cdot c \cdot \sigma) \cap \Z^n$, as needed. So the refinement $\bar{U}$ has competitive equilibrium. 
\hfill\qed
\endproof
\medskip 

This theorem should be not be confused with the statement that the linear program in (\ref{eqn:primal}) has a rational solution. The later guarantees that to each fixed point $a \in \conv_\Z(A)$, one can find a $c_a \in \mathbb{N}$ such that the $c_a$-refined auction has competitive equilibrium at $a$. Applied over all points $a \in \conv_\Z(A)$, this says that for an appropriate $c^\ast \in \mathbb{N}$, the auction has competitive equilibrium for all points in $c^\ast \cdot \conv_\Z(A)$. However, this set can be strictly smaller than $\conv_\Z(c^\ast \cdot A)$, the support of the $c^\ast$-refinement of the auction. Therefore, the $c^\ast$-refined auction may not have competitive equilibrium. 

\subsection{Theorem \ref{thm:n2} in higher dimensions and the Oda Conjecture}
\label{sec:conjecture}
To generalize Theorem~\ref{thm:n2} or the more modest Corollary \ref{cor:n2.stable} to dimensions $n \geq 3$, we need sufficient conditions to ensure that polytopes have IDP. Unfortunately, except for some limited cases, there are no known simple criterion. The most famous conjecture in this direction is the \emph{smooth polytope conjecture}. A full-dimensional lattice polytope in $\RR^n$ is called {\em smooth} if at every vertex, the primitive integral edge directions are linearly independent and unimodular. Smoothness here does not refer to `roundness', but rather to the fact that the toric variety defined by such polytopes are smooth.  Up to unimodular transformations there are finitely many smooth polytopes in a given dimension containing a given number of lattice points, see \cite{smooth}.

\begin{conjecture}[Smooth Polytope Conjecture]
Smooth polytopes have IDP.
\end{conjecture}

In proving Theorem \ref{thm:n2}, we relied on \cite[Theorem 1]{HNPS}. Its counterpart in higher dimension, which generalizes the Smooth Polytope Conjecture, is the Oda Conjecture below. For further discussions and recent progress on this conjecture, see~\cite{lason2014non}.

\begin{conjecture}[Oda Conjecture]
In any dimension $n \geq 1$, equation (\ref{eqn:normal}) is true if one polytope $F_i$ is smooth and its normal fan refines the normal fan of $F_{j}$ for all $j=1,\dots,J$. 
\end{conjecture}

We restate the Oda conjecture in terms of the stable auction below. Obvious adaptation of the proof of Theorem \ref{thm:n2} shows these two conjectures are equivalent. In particular, any other sufficient condition for competitive equilibrium in stable product-mix auctions is equivalent to weaker forms of the Oda conjecture.

\begin{conjecture}[Oda Conjecture in terms of stable auctions]\label{conj:n3}
Suppose a product-mix auction with aggregate valuation $U$ is stable and that each maximal face in the regular subdivision~$\Delta_U$ is smooth. 
Then the product-mix auction has competitive equilibrium.
\end{conjecture}

\bibliographystyle{plain}
\bibliography{productMixAuction}

\end{document}